\documentclass[runningheads]{llncs} 

\usepackage[utf8]{inputenc}
\usepackage[T1]{fontenc}      % Needed for accented characters
\usepackage{amsmath}
\usepackage{amssymb}
\usepackage{amsfonts}
\usepackage{cancel}
\usepackage{caption}
\usepackage{hyperref}
\usepackage{color}
\usepackage{mathtools}
\usepackage[inline]{enumitem}
\usepackage{tabularx, array}
\usepackage{tikz}
\usetikzlibrary{patterns, intersections, positioning, shapes.geometric, calc,
automata, arrows,decorations.markings}

\usepackage{subcaption}
\usepackage{tcolorbox}
\usepackage{xspace}
\usepackage{xcolor}
\usepackage{mathrsfs}
\newcommand{\calC}{\mathcal{C}}

\newcommand{\PhiLoc}{\Phi_\locations}

\newcommand{\state}{s}
\newcommand{\startingstate}{\state_0}
\newcommand{\states}{S}
\newcommand{\clocks}{\calC}
\newcommand{\clock}{c}
\newcommand{\const}{\mathbf{K}}
\newcommand{\y}{\wedge}

\newcommand{\location}{\ell}
\newcommand{\startinglocation}{\ell_0}
\newcommand{\locations}{L}
\newcommand{\Act}{\mathit{Act}}
\newcommand{\Acti}{\inputlabs}
\newcommand{\Acto}{\outlabs}

\newcommand{\project}[1]{[#1]\downarrow}

\newcommand{\ioco}{\textbf{ioco}\xspace}
\newcommand{\nioco}{\ \not\!\!\!\!\!\ioco}

\newcommand{\tioco}{\textbf{tioco}\xspace}
\newcommand{\pass}{\textbf{pass}\xspace}
\newcommand{\fail}{\textbf{fail}\xspace}

\newcommand{\test}{\textbf{t}\xspace}

\newcommand{\atest}{\test}
\newcommand{\testsuite}{\textbf{T}\xspace}
\newcommand{\testLTS}{\testsuite_{\mathit{LTS}}\xspace}
\newcommand{\atestLTS}{\testsuite_{\mathit{LTS}}\xspace}
\newcommand{\testTA}{\test_{\mathit{TA}}\xspace}
\newcommand{\atestTA}{\testsuite_{\mathit{TA}}\xspace}
 
\newcommand{\trans}{\rightarrow}
\newcommand{\Trans}{\Rightarrow}
\newcommand{\transition}[1]{\xrightarrow{#1}}
\newcommand{\ntransition}[1]{\not\xrightarrow{#1}}
\newcommand{\lntransition}[1]{\not\!\!\!\xrightarrow{#1}}

\newcommand{\ltstuple}{\langle \states, \Act, \trans, \startingstate \rangle}

\newcommand{\tatuple}{\langle \locations, \Act, \PhiLoc, \calC, \rightarrow, \location_0\rangle}

\newcommand{\taification}{\chi^{\scriptstyle M}\!}

\definecolor{lightgreen}{RGB}{224,255,224}
\definecolor{darkgreen}{RGB}{0,128,0}
\definecolor{lightblue}{RGB}{200,200,255}
\definecolor{darkblue}{RGB}{56,56,255}
\definecolor{lightred}{RGB}{200,200,255}
\definecolor{darkred}{RGB}{139,0,0}

\newcommand{\RRplus}{\mathbb{R}_{>0}}
\newcommand{\RRnonzero}{\ensuremath{\mathbb{R}_{\geq 0}}}
\newcommand{\labs}{\Act}
\newcommand{\Actq}{\Act^\delta}
\newcommand{\qAct}{\Actq}
\newcommand{\outlabs}{\labs_O}

\newcommand{\outlabsdelta}{\labs_O^\delta}%\labs_{O \delta}}
\newcommand{\olabsd}{\outlabsdelta}

\newcommand{\Actoq}{\olabsd}
\newcommand{\inputlabs}{\labs_I}
\newcommand{\trace}{\sigma}
\newcommand{\ttrace}{\rho}
\newcommand{\system}{\mathcal{A}}
\newcommand{\sys}{\system}
\newcommand{\systemb}{\mathcal{B}}
\newcommand{\systemc}{\mathcal{C}}
\newcommand{\systemd}{\mathcal{D}}

\newcommand{\systemTA}{\mathscr{A}}
\newcommand{\sysTA}{\systemTA}

\newcommand{\impl}{\system_I}
\newcommand{\implTA}{\systemTA_I}
\newcommand{\spec}{\system_S}
\newcommand{\specTA}{\systemTA_S}

\newcommand{\qsysb}{\systemb^\delta}

\newcommand{\action}{a}
\newcommand{\iaction}{\inputlab}
\newcommand{\oaction}{\outlab}
\newcommand{\ooaction}{\ooutlab}
\newcommand{\inputlab}{i}
\newcommand{\iinputlab}{i$'$}
\newcommand{\outlab}{o}
\newcommand{\ooutlab}{o$'$}

\newcommand{\traces}{\mathit{traces}}
\newcommand{\ttraces}{\mathit{ttraces}}
\newcommand{\Straces}{\mathit{Straces}}

\newcommand{\after}{\;\textbf{after}\;}
\newcommand{\afterM}{\;\mathbf{after}_M\;}
\newcommand{\iocorel}{\mathbf{ioco}}

\newcommand{\tiocorel}{\mathbf{tioco}}
\newcommand{\tiocoM}{\mathbf{tioco_M}}

\newcommand{\out}{{\bf{out}}}
\newcommand{\inp}{{\bf{in}}}

\newcommand{\va}[1]{\overset{#1}{\rightarrow}}
\newcommand{\nva}[1]{\overset{#1}{\not\rightarrow}}
\newcommand{\vaa}[0]{\rightarrow}
\newcommand{\RRo}{\RRnonzero}

\newcommand{\vva}[1]{\overset{#1}{\longrightarrow}}

\newcommand{\outM}{{\bf{out}}_M}

\newcommand{\SttracesM}{\mathit{Sttraces_M}}
\newcommand{\ttracesM}{\mathit{
ttraces_M}}

\newcommand{\Sttraces}{\mathit{Sttraces}}

\newcommand{\tim}{{\mathit{d}}}
\newcommand{\tq}{{\ |\ }}
\newcommand{\al}[1]{^{#1}}

\newcommand{\tikztrans}[3]{
  \begin{tikzpicture}[node distance=2cm]
    \tikzset{state/.style= {draw, circle}} 
     \node[state] (0) {};
     \node (1) [right of=0] {};
     \path[->]
    (0) edge node [above] {#1} node [below, text width=12mm, align=center] {#2 #3} (1)
    ;
    \end{tikzpicture}}
\newcommand{\tikztransinv}[4]{
  \begin{tikzpicture}[node distance=2cm]
    \tikzset{timedstate/.style= {draw, rounded corners}}
     \node[timedstate] (0) {#4};
     \node (1) [right of=0] {};
     \path[->]
    (0) edge node [above] {#1} node [below, text width=12mm, align=center] {#2 #3} (1)
    ;
    \end{tikzpicture}}

\title{Time for Quiescence:\\
Modelling quiescent behaviour in testing via time-outs in timed automata}
\titlerunning{Time for Quiescence}
\author{
Laura Brand\'{a}n Briones \inst{1} \orcidID{0009-0001-7402-0077} \and\\
Marcus Gerhold \inst{2} \orcidID{0000-0002-2655-9617} \and 
Petra van den Bos  \inst{2} \orcidID{0000-0002-9212-1525} \and 
Mari{\"e}lle Stoelinga  \inst{2,3} \orcidID{0000-0001-6793-8165} \\
}
\authorrunning{L. Brand\'{a}n Briones et al.}

\institute{
Universidad Nacional de C\'ordoba, C\'ordoba, Argentina \email{laura.brandan@unc.edu.ar} \and
University of Twente, Enschede, The Netherlands 
\email{\{m.gerhold,p.vandenbos,m.i.a.stoelinga\}@utwente.nl} \and
Radboud University, Nijmegen, the Netherlands
}

\begin{document}
\maketitle
\begin{abstract}
    Model-based testing (MBT) derives test suites from a behavioural specification of the system under test. 
    In practice, engineers favour simple models, such as labelled transition systems (LTSs). 
    However, to deal with \emph{quiescence}---the absence of observable output---in practice, a time-out needs to be set to conclude observation of quiescence.
    Timed MBT exists, but it typically relies on the full 
    arsenal of timed automata (TA).
    
    We present a lifting operator $\taification$ that adds timing without the TA overhead: given an LTS, $\taification$ introduces a single clock for a user chosen time bound $M>0$ to declare quiescence. 
    In the timed automaton, the clock is used to model that outputs should happen before the clock reaches value $M$, while quiescence occurs exactly at time $M$.
    This way we provide a formal basis for the industrial practice of choosing a time-out to conclude quiescence.
    Our contributions are threefold:
    (1) an implementation conforms under $\iocorel$ if and only if its lifted version conforms under timed $\tiocoM$ (2) applying $\taification$ before or after the standard \ioco test-generation algorithm yields the same set of tests, and (3) the lifted TA test suite and the original LTS test suite deliver identical verdicts for every implementation.

\end{abstract}
%- - - - - - - - Sections - - - - - - - - 
\section{Introduction}

\paragraph{Model-based testing.} Model-based testing is an effective way of testing, by providing automated test generation, execution, and evaluation. Central in model-based testing is a system specification model $\spec$ that pins down exactly the desired system behaviour. 
Test cases are derived automatically from $\spec$ and executed automatically against the system under test (SUT). Then, a test verdict (pass/fail) is given: if running a test case against the SUT exhibits a behaviour that adheres to $\spec$, the verdict pass is given; otherwise, a fail verdict is given.

\paragraph{The input-output conformance framework.}
Mathematical correctness in MBT is expressed through formal conformance relations that compare an implementation with its specification.
A well-studied instance is \emph{input–output conformance} (\ioco)~\cite{T2008}. In this paper, we build on $\ioco$ theory, as it is a widely recognized model-based testing framework applied in both academic research and industrial practice.
Ioco-theory provides a rigorous mathematical underpinning of model-based testing, in terms of {\em soundness} (ensuring that all emitted test verdicts are correct) and {\em completeness} (ensuring that the test generation algorithms have no inherent blind spots, i.e. all potential non-conformities can in principle be found by the test generation algorithms).

More precisely, the $\ioco$-conformance relation pins down exactly when an implementation model $\impl$ correctly implements a specification model $\spec$, where both $\impl$ and $\spec$ are given as labelled transition systems, under the required assumptions. 

\paragraph{Quiescence.} An intricate aspect in model-based testing is the handling of {\em quiescence}: what happens if the SUT does not provide any output? If quiescent behaviour (i.e. absence of output) is allowed, then the test should lead to the verdict pass. Otherwise, if  quiescence is disallowed, the test should yield the verdict fail. 
To answer this question, the system specification must be augmented with quiescence information. 
To answer this question, the observation semantics of the specification must make quiescence explicit. In ioco-theory, this can be done via \emph{suspension-traces}: a trace records the special action $\delta$ whenever it reaches a state that has no outgoing outputs.

In industrial practice, quiescence is handled with a time-out: if no output is observed within a fixed bound $M$, the system is deemed quiescent and no further reaction is expected.
These time-outs can be naturally modelled as a timed automaton. Interestingly, several timed variants of \ioco have been proposed, including tioco~\cite{LMN2004}, rtioco~\cite{KT2004}, and $\tiocoM$~\cite{BBB04}.
In this paper we show that the two frameworks of \ioco  and $\tiocoM$ are closely related.

That is, if we model quiescence via a time-out in a timed automaton, then conformance is preserved. 
More technically, we translate each labelled transition system $\system$ into a timed automaton $\taification(\system)$ which ensures that each output occurs before $M$ time units. 
Then we show that if an LTS $\impl$ conforms to a specification $\spec$ (i.e. $\impl\ \ioco\ \spec$), 
then this is also the case after transforming these into timed automata, 
via the timed conformance relations (i.e. $\taification(\impl) \; \tioco_M \; \taification(\spec)$). Note that the  $\tioco_M$-relation is parametrized by the time-out constant $M$.

Additionally, we show that the correspondence between LTS and TA also holds on the level of test cases: since test cases are LTSs the operation $\taification$ can also be applied to them in such a way that applying $\test$ to $\impl$ yields the same verdict as applying $\taification(\test)$ to $\taification(\impl)$.
Moreover, we show that $\taification$ commutes with test generation: applying $\taification$ before or after the standard ioco
test-generation algorithm yields the same set of tests.

We note that with our theory, practitioners only need to model the system’s input–output behaviour as an LTS, and can then apply the transformation to automatically obtain a timed automaton with appropriate time-outs. Practitioners thus do not need to worry about how to implement time-outs: the transformation ensures that outputs are only allowed before the time-out, that inputs are provided before the time-out as well (no unnecessary waiting needed), and that quiescence is concluded directly after waiting $M$ time units.

\paragraph{Paper overview:}
\autoref{sec:background} recalls prerequisites on labelled transition systems and \ioco, and~\autoref{sect:ta} the necessary theory of timed automata and $\tiocoM$.
In \autoref{sec:transformation} we define the transformation rules that translates from LTS into a TA  and prove preservation of conformance from \ioco to $\tiocoM$ \emph{(contribution 1)}.
In \autoref{sec:testing} we define test cases for LTS and TA, and prove commutativity of $\taification$ when applied on LTS and all its test cases \emph{(contribution 2)}, and that an implementation that fails a test case in the LTS setting will also do so in the TA setting \emph{(contribution 3)}.
Lastly, we discuss related work in \autoref{sec:related-work} and conclude the paper in \autoref{sec:conclusion}.

\section{Labelled Transition Systems and $\ioco$}
\label{sec:background}

Labelled transition systems (LTSs) are standard transition systems with labelled transitions and states. For the purposes of model-based testing, LTSs specify the behaviour of a System Under
Test (SUT). Here it is common to distinguish between inputs labels and output labels of the system explicitly. Input labels denote actions provided to the SUT, and output labels denote output actions performed by the SUT. For simpler definitions and notations, we choose to leave the internal, unobservable $\tau$ label out of scope of this paper.

\begin{definition}[Labelled Transition System]
\label{def:lts}
A labelled transition system (LTS) is a 4-tuple $\system = \ltstuple$, where $\states$ is a finite set of states with a unique starting state $\startingstate\in \states$, $\Act$ is the finite set of actions partitioned into input and output actions, i.e. $\Act= \Acti\sqcup\Acto$, and $\trans\ \subseteq\states\times\Act\times\states$ is the transition relation.

\begin{itemize}\itemsep0mm 
    \item We write $\state\transition{\action}\state'$ for $(\state,\action,\state')\in\ \trans$, and $\state\transition{\action}$ if $\state\transition{\action}\state'$ for some $\state'\in \states$ and   $\state\ntransition{\action}$ if no such $\state'$ exists
    
    \item If $\trace = \action_1\ldots \action_n$ for $\action_1,\ldots,\action_n\in\Act$, we write $\state\transition{\trace} \state'$ if there are states $\state_1, \ldots, \state_{n-1}\in\states$ such that $\state\transition{\action_1} \state_1 \ldots \state_{n-1}\transition{\action_n} \state'$. We call $\trace\in \Act^*$ a trace
    
    \item We write $\traces(\state) = \{\trace\in\Act^* \mid\ \state\transition{\trace}\}$ and $\traces(\system)=\traces(\startingstate)$
    
    \item We let $\sqsubseteq$ denote the prefix relation on traces.  If $\trace,\trace'\in \Act^*$, $ \sigma =\action_1\ldots\action_n$ and            $\trace'=\action_1\ldots\action_i$ for some $i\leq n$ we write $\trace'\sqsubseteq\trace$ to denote $\sigma'$ as subtrace of $\trace$.

        \end{itemize}
\end{definition}

Throughout the paper, we use the suffix $-?$ for inputs and $-!$ for outputs, but these suffixes are not technically part of the label itself.
Also, we use $\inputlab?$ and $\iinputlab?$ to denote inputs, and $\outlab!$ and $\ooutlab!$ for outputs.

Besides using an LTS as a specification of expected behaviour for model-based testing, we assume that the actual behaviour of the SUT, the implementation, can be represented as an LTS as well. 
In \ioco theory, the implementation is required to accept all input at all times, i.e. each of its states is \emph{input-enabled}.
We call an input-enabled LTS an \emph{input-output transition system} (IOTS).
\begin{definition}[Input-Output Transition System]
\label{def:iots}
An input-output transition system (IOTS) is an input-enabled LTS, i.e. 
for all $\inputlab? \in \Acti$ and all $\state \in \states$, we have $\state\transition{\inputlab?}$.

\end{definition}
\begin{example}
    \begin{figure}
    \hspace{-1cm}
    \centering
    \begin{subfigure}[b]{0.45\textwidth}
    \begin{tikzpicture}
    \tikzset{state/.style= {draw, circle}} 
    \tikzset{timedstate/.style= {draw, rounded corners}}
    \node[state] (s0) {};
    \node[above left of=s0] (init) {};
    \node[state, below right of=s0, xshift=0.75cm] (s1) {};
        \node[state, right of=s1, xshift=0.75cm] (s3) {};
    \node[state, above right of=s0, xshift=0.75cm] (s2) {};
        \node[state, right of=s2, xshift=0.75cm] (s4) {};
            \node[state, below right of=s4, xshift=0.75cm] (s5) {};
            \node[state, above right of=s4, xshift=0.75cm] (s6) {};

    \draw[->] (init) to (s0);
    \draw[->] (s0) to node[below left, xshift=0.1cm] {\iaction?} (s1);
    \draw[->] (s0) to node[above left, xshift=0.1cm] {\iaction?} (s2);
    \draw[->] (s1) to node[below] {\oaction!} (s3);
    \draw[->] (s2) to node[above] {\iaction?} (s4);
    \draw[->] (s4) to node[below left, xshift=0.3cm] {\ooaction!}(s5);
    \draw[->] (s4) to node[above left, xshift=0.2cm] {\oaction!} (s6);
    
    \end{tikzpicture}
    \subcaption{LTS $\system$}
    \end{subfigure}
    \quad\quad
    \begin{subfigure}[b]{0.55\textwidth}
    \begin{tikzpicture}
    \tikzset{state/.style= {draw, circle}} 
    \tikzset{timedstate/.style= {draw, rounded corners}}

    \node[state] (s0) {};
    \node[above left of=s0] (init) {};
    \node[state, below right of=s0, xshift=0.75cm] (s1) {};
        \node[state, right of=s1, xshift=0.75cm] (s3) {};
    \node[state, above right of=s0, xshift=0.75cm] (s2) {};
        \node[state, right of=s2, xshift=0.75cm] (s4) {};
            \node[state, below right of=s4, xshift=0.75cm] (s5) {};
            \node[state, above right of=s4, xshift=0.75cm] (s6) {};

    \draw[->] (init) to (s0);
    \draw[->] (s0) to node[below left, xshift=0.1cm] {\iaction?} (s1);
    \draw[->] (s0) to node[above left, xshift=0.1cm] {\iaction?} (s2);
    \draw[->] (s1) to node[below] {\oaction!} (s3);
    \draw[->] (s2) to node[above] {\iaction?} (s4);
    \draw[->] (s4) to node[below left, xshift=0.2cm] {\ooaction!}(s5);
    \draw[->] (s4) to node[above left, xshift=0.3cm] {\oaction!} (s6);

    \draw[->, red, densely dashed] (s1) to [loop above] node {\iaction?} (s2);
    \draw[->, red, densely dashed] (s3) to [loop above] node {\iaction?} (s3);
    \draw[->, red, densely dashed] (s4) to [loop above] node {\iaction?} (s4);
    \draw[->, red, densely dashed] (s5) to [loop above] node {\iaction?} (s5);
    \draw[->, red, densely dashed] (s6) to [loop above] node {\iaction?} (s6);
    \end{tikzpicture}
    \subcaption{IOTS $\systemb$}
    \end{subfigure}
    \caption{\label{fig:lts-iots}
        LTS $\system = \langle\states,\{\iaction?,\oaction!,\ooaction!\},\rightarrow_\system,s_0\rangle$ and IOTS $\systemb = \langle\states,\{\iaction?,\oaction!,\ooaction!\},\rightarrow_\systemb,s_0\rangle$. 
    }
\end{figure}

    Figure~\ref{fig:lts-iots} shows an example LTS $\system$, and IOTS $\systemb$, such that $\systemb$ is an input-enabled version of $\system$.
    The newly added input transitions are highlighted.
\end{example}
Traces represent the \emph{visible} behaviour of a labelled transition system.
A state that does not have any outgoing output transitions is \emph{quiescent}.
Quiescent behaviour of a state is made explicit with action $\delta$ to the same state. 

\begin{definition}[Quiescence]
    A state $\state \in \states$ is \emph{quiescent} iff $\forall\ \outlab! \in \Acto: \state\ntransition{\outlab!}$.
    We write $\qAct = \Act \cup \{\delta\}$, and $\outlabsdelta = \outlabs\cup\{\delta\}$.
\end{definition}

\begin{example}
    Figure~\ref{fig:delta-iots} shows $\systemb$, the IOTS from \autoref{fig:lts-iots}(b), where the quiescent states are highlighted by $\delta$-self-loops.
    \begin{figure}
    \centering
    \begin{tikzpicture}
    \tikzset{state/.style= {draw, circle}} 
    \tikzset{timedstate/.style= {draw, rounded corners}}

    \node[state] (s0) {};
    \node[above left of=s0] (init) {};
    \node[state, below right of=s0, xshift=0.75cm] (s1) {};
        \node[state, right of=s1, xshift=0.75cm] (s3) {};
    \node[state, above right of=s0, xshift=0.75cm] (s2) {};
        \node[state, right of=s2, xshift=0.75cm] (s4) {};
            \node[state, below right of=s4, xshift=0.75cm] (s5) {};
            \node[state, above right of=s4, xshift=0.75cm] (s6) {};

    \draw[->] (init) to (s0);
    \draw[->] (s0) to node[below left] {\iaction?} (s1);
    \draw[->] (s0) to node[above left] {\iaction?} (s2);
    \draw[->] (s1) to node[below] {\oaction!} (s3);
    \draw[->] (s2) to node[above] {\iaction?} (s4);
    \draw[->] (s4) to [loop above] node {\iaction?} (s4);
    \draw[->] (s4) to node[below left] {\oaction!}(s5);
    \draw[->] (s4) to node[above left] {\ooaction!} (s6);
    
    \draw[->, red, densely dashed] (s0) to [loop above] node {$\delta$} (s0);
    \draw[->] (s1) to [loop above] node {\iaction?} (s2);
    \draw[->, red, densely dashed] (s3) to [loop right] node {$\delta$} (s3);
    \draw[->] (s3) to [loop above] node {\iaction?} (s3);
    \draw[->, red, densely dashed] (s2) to [loop above] node {$\delta$} (s2);
    \draw[->, red, densely dashed] (s5) to [loop right] node {$\delta$} (s5);
    \draw[->, red, densely dashed] (s6) to [loop right] node {$\delta$} (s6);
    \draw[->] (s5) to [loop above] node {\iaction?} (s5);
    \draw[->] (s6) to [loop above] node {\iaction?} (s6);
    \end{tikzpicture}
    \caption{System $\qsysb$ with its quiescent states highlighted by $\delta$-self-loops.}
    \label{fig:delta-iots}
\end{figure}
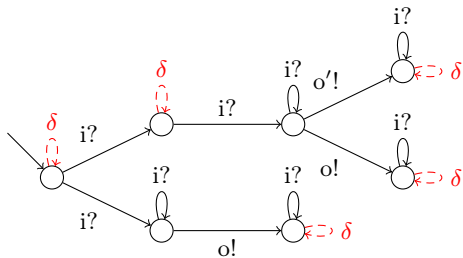

\end{example}

Suspended traces extend regular traces with the special quiescence action $\delta$ in quiescent states.
These suspended traces are then used to define the \ioco\ relation~\cite{T2008}, where $\delta$ is considered among the outputs that an LTS/IOTS may produce. \autoref{rem:ioconotation} introduces the necessary notation for the definition of \ioco\ in \autoref{def:ioco}. Since we left $\tau$ out of scope of this paper, we can define the $\after$ function directly from the transitions and quiescent states of an LTS.

\begin{definition}
[\ioco notation]
\label{rem:ioconotation}
Let $\system=\ltstuple$ be an LTS. 
Below are \ioco specific notations:
\begin{itemize}
    \item All output actions, including $\delta$, enabled in a state $\state \in \states$ are:\\
    $\out(\state)= \{\oaction \in \outlabs\tq \state \transition{\oaction}\} \cup \{\delta\tq \text{if } \state\text{ is quiescent}\}$
    
    \item All input actions enabled in a state $\state \in \states$ are: 
    $\inp(\state) = \{\iaction \in \Acti\tq \state \transition{\iaction}\}$

     \item Let $\varepsilon$ be the empty sequence, $\action \in \qAct$ an action or $\delta$, and $\trace \in (\qAct)^*$ a sequence of actions, including $\delta$. Then states after these sequences, starting in $\state \in \states$, are:
    \begin{align*}
    \state \after \varepsilon &= \{ \state \}\\
        \state \after \action &= \{\state' \mid \action \in \Act \wedge\ \state\transition{\action} \state'\} \cup \{\state \mid \action = \delta \wedge\ \state \text{ is quiescent} \}\\
        \state \after \action\,\trace &= \bigcup\ \{\state' \after \trace \mid  \state'\in \state \after \action\}
    \end{align*}

    \item The suspended traces of a state $\state \in \states$, i.e. traces including quiescence, are:\\
    $\Straces(\state) = \{\trace \in (\qAct)^* \mid \state \after \trace \neq \emptyset \}$
    \item For an LTS $\system$ we write:  $\system\after\trace=\startingstate\after\trace$, and $\Straces(\system) = \Straces(\startingstate)$.
    
\end{itemize}
\end{definition}

We are now ready to recall classic input-output conformance (\ioco) \cite{T2008}. This relation defines when an IOTS implementation conforms to an LTS specification. 
\begin{definition}[ioco]
\label{def:ioco}
Let $\impl$ be an IOTS and $\spec$ an LTS. Then
$\impl\ \ioco\ \spec$ iff
$$
\forall\ \trace \in \Straces(\spec) : \out(\impl \after \trace) \subseteq \out(\spec \after \trace).
$$

\end{definition}

\begin{example}
    Figure~\ref{fig:ioco} presents two systems, on the left we have one implementation $\systemc \in$ IOTS and on the right we have another implementation $\systemd \in$ IOTS of the specification $\system \in$ LTS presented in Figure~\ref{fig:lts-iots}(a). We can observe that $\systemc\ \ioco\ \system$ but $\systemd\ \nioco\ \system$, because of the suspended trace $\trace = \iaction? \cdot \delta \cdot \iaction?$, where $\delta \in \out(\systemd \after \trace)$ but $\delta \not\in \out(\system \after \trace)$. 
    \begin{figure}
    \hspace{-1cm}
    \centering
    \begin{subfigure}[b]{0.45\textwidth}
    \begin{tikzpicture}
    \tikzset{state/.style= {draw, circle}} 
    \tikzset{timedstate/.style= {draw, rounded corners}}
    \node[state] (s0) {};
    \node[above left of=s0] (init) {};
    \node[state, below right of=s0, xshift=0.75cm] (s1) {};
        \node[state, right of=s1, xshift=0.75cm] (s3) {};
    \node[state, above right of=s0, xshift=0.75cm] (s2) {};
        \node[state, right of=s2, xshift=0.75cm] (s4) {};
            \node[state, above right of=s4, xshift=0.75cm] (s6) {};

    \draw[->] (init) to (s0);
    \draw[->] (s0) to node[below left, xshift=0.1cm] {\iaction?} (s1);
    \draw[->] (s0) to node[above left, xshift=0.1cm] {\iaction?} (s2);
    \draw[->] (s1) to node[below] {\oaction!} (s3);
    \draw[->] (s2) to node[above] {\iaction?} (s4);
    \draw[->] (s4) to node[above left, xshift=0.2cm] {\oaction!} (s6);
    \draw[->] (s1) to [loop above] node {\iaction?} (s1);
    \draw[->] (s3) to [loop above] node {\iaction?} (s3);
    \draw[->] (s4) to [loop above] node {\iaction?} (s4);    
    \draw[->] (s6) to [loop above] node {\iaction?} (s6);    
    
    \end{tikzpicture}
    \subcaption{IOTS $\systemc$}
    \end{subfigure}
    \quad\quad
    \begin{subfigure}[b]{0.55\textwidth}
    \begin{tikzpicture}
    \tikzset{state/.style= {draw, circle}} 
    \tikzset{timedstate/.style= {draw, rounded corners}}

    \node[state] (s0) {};
    \node[above left of=s0] (init) {};
    \node[state, below right of=s0, xshift=0.75cm] (s1) {};
        \node[state, right of=s1, xshift=0.75cm] (s3) {};
    \node[state, above right of=s0, xshift=0.75cm] (s2) {};
        \node[state, right of=s2, xshift=0.75cm] (s4) {};

    \draw[->] (init) to (s0);
    \draw[->] (s0) to node[below left, xshift=0.1cm] {\iaction?} (s1);
    \draw[->] (s0) to node[above left, xshift=0.1cm] {\iaction?} (s2);
    \draw[->] (s1) to node[below] {\oaction!} (s3);
    \draw[->] (s2) to node[above] {\iaction?} (s4);
    
    \draw[->] (s1) to [loop above] node {\iaction?} (s4);
    \draw[->] (s3) to [loop above] node {\iaction?} (s3);
    \draw[->] (s4) to [loop above] node {\iaction?} (s4);
    \end{tikzpicture}
    \subcaption{IOTS $\systemd$}
    \end{subfigure}
    \caption{\label{fig:ioco}
        IOTS $\systemc = \langle\states,\{\iaction?,\oaction!,\ooaction!\},\rightarrow,s_0\rangle$ and IOTS $\systemd = \langle\states,\{\iaction?,\oaction!,\ooaction!\},\rightarrow,s_0\rangle$. 
    }
\end{figure}

\end{example}

\section{Timed Automata and $\tiocoM$}\label{sect:ta}

Inspired from~\cite{A1999} we define a timed automaton (TA) as a labelled transition system with clock
variables, clock constraints on states (now called locations), and transitions with clock guards.
Again, we split actions into inputs and outputs.
We first define notation for clock constraints.

\begin{definition}[Clock constraint]
\label{def:clockconstraint}
    The set of clock constraints over a set of clock variables $\clocks$ is $\Phi(\clocks)$, defined as in \cite{A1999}. In particular a \emph{clock constraint} $\pi \in \Phi(\clocks)$ is an element from grammar:
    $$\phi := \clock \le \const \mid \const \le \clock \mid \clock < \const \mid \const < \clock \mid \phi_1 \wedge \phi_2$$ for any clock variables $\clock \in \clocks$ and constants $\const \in \RRnonzero$.
\end{definition}
\begin{definition}[Timed Automaton]
A \emph{timed automaton} (TA) $\systemTA$ is a tuple $\tatuple$, where
\begin{itemize}\itemsep0mm
    \item $\locations$ is a finite set of locations with $\startinglocation\in \locations$ as the initial location
    
    \item $\Act$ is the finite set of action labels subdivided in input and output actions, i.e. $\Act= \Acti\sqcup\Acto$
    
    \item $\clocks$ is a finite set of clock variables
    
    \item $\PhiLoc: \locations \rightarrow \Phi(\clocks)$ is a function that maps each location $\ell\in \locations$ to some clock constraint in $\Phi(\calC)$. We call $\PhiLoc(\ell)$ the \emph{invariant} of $\ell$
    \item $\rightarrow\ \subseteq \locations\times \Act \times \Phi(\clocks)\times 2^{\clocks}\times\locations$ is a set of timed transitions.
    A timed transition $\langle \ell, \action, \phi, \lambda, \ell'\rangle \in\ \rightarrow$ represents an edge from location $\ell$ to location $\ell'$, on label $\action$. With $\phi \in \Phi(\clocks)$ we denote its clock constraint, called \emph{guard}, which specifies when the transition is enabled. The set $\lambda\subseteq \clocks$ gives the clocks to be reset, i.e. set to value 0, with this transition.
\end{itemize}
\end{definition}
As per usual, \emph{locations} in a TA differ from \emph{states}.
Commonly, the latter is only used when we talk about the TA's semantics, i.e. a location is a tuple of a state and a clock evaluation.
Consequently, a TA's semantics has uncountably many states based on the uncountably many non-negative clock valuations.
As is common, we exclude \emph{Zeno} behaviour, i.e. infinitely many actions happening in a finite amount of time.

Traces in timed automata are sequences of $(d,\action)$-tuples, denoting that $\action$ happens after the passage of $d$ time units in the previous location. The invariants and guards of the TA specify whether the transition can be taken after $d$ time.
\begin{definition}[TA Notation] 
\label{def:tanotation}
Let $\systemTA=\tatuple$ be a TA:
\begin{itemize}\itemsep0mm    
    \item We write $\location\transition{(d,\action)}\ell'$ for $(d,a) \in \RRnonzero\times\Act$ if there is a $\langle \location, \action, \phi, \lambda, \ell'\rangle\in\ \trans$ such that $\phi$ and $\PhiLoc(\location)$ are true for time $d$ that is spent between $\ell$ and $\ell'$, and such that $\PhiLoc(\ell')$ is true after updating the clocks with the resets from $\lambda$ 
    \item Timed traces are sequences of non-negative numbers and visible actions, i.e. $\ttraces(\ell) = \{\ttrace \in(\RRnonzero\times \Act)\al{*} \tq  \ell \transition{\ttrace}\}$ 
    \item Other notations on traces and subtraces carry over straightforwardly from \autoref{def:lts}.
    
\end{itemize}
\end{definition}

To define $\tiocoM$~\cite{BBB04}, the timed variant of \ioco, we introduce some notation additionally. We explicitly use a duration parameter $M\in\RRplus$ after which we declare quiescence.

\begin{definition}[$\tiocoM$\ notation] 
\label{rem:tioconotation}
Let $\systemTA=\tatuple$ be a TA, and let $M\in\RRplus$ as the time to declare quiescence explicitly:
\begin{itemize}
    \item A location $\location \in \locations$ is \emph{quiescent} iff\\\text{  } $\forall\ \tim \in \RRnonzero, \forall\ \outlab! \in \Acto: d < M \implies \location\ \lntransition{(\tim,\outlab!)}$
    \item  The outputs or inputs enabled in a location $\location$ are given as:
    \begin{align*}
    \outM(\ell)=& \{(\tim,\oaction!)\in\RRnonzero\times\Acto\tq  \location \transition{(\tim,\oaction!)}\} \cup \{(M,\delta)\tq \ell \text{ is quiescent}\}\\
    \out(\ell) =& \{\oaction! \in \Actoq \tq \exists\ (\tim, \oaction!) \in \outM(\ell) \}\\
    \inp(\ell) =& \{\iaction? \in\Acti\tq \exists\ (\tim, \iaction?) \in \RRnonzero\times\Acti \y \location\transition{(\tim,\iaction?)}\}
    \end{align*}
    \item We define locations after a sequence of tuples of a time duration $d$ and action $a$, including the tuple $(M,\delta)$ for quiescent locations, starting in location $\ell$:
    \begin{align*}
    \ell \afterM \epsilon &= \{ \ell \}\\
        \ell \afterM (\tim,\action) &= \{\ell' \mid \action \in \Act \wedge\ \ell\transition{(\tim,\action)} \ell'\} \cup \{\ell \mid (\tim,\action) = (M,\delta) \wedge\ \ell \text{ is quiescent} \}\\
        \ell \afterM (\tim,\action)\ttrace &= \bigcup\ \{\ell' \afterM \ttrace \mid  \ell'\in \ell \afterM (\tim,\action)\}
    \end{align*}
    \item We define the suspended timed traces as the traces of location $\ell$, including $\delta$ at time $M$, for quiescent locations encountered in the trace:
    $ \SttracesM(\ell) = \{\ttrace \in (\RRnonzero \times \qAct)^* \mid \ell \afterM \ttrace \neq \emptyset \}$
    \item We write:  $\systemTA\after\ttrace=\ell_0\after\ttrace$ and $\SttracesM(\systemTA) = \SttracesM(\ell_0)$.
\end{itemize}
\end{definition}
Lastly, for testing purposes, we define input enabled TAs called input-output timed automata (IOTA).
\begin{definition}[Input-Output Timed Automaton]
An input-output timed automaton (IOTA) is an input-enabled TA for parameter $M\in\RRplus$, i.e.\\
$$\forall\ d < M  \in \RRnonzero, \forall\ \iaction? \in \Acti, \forall\ \location \in \locations: \location\transition{(d,\iaction?)}.$$
\end{definition}

We will now define timed input-output conformance $\tiocoM$, as in \cite{BB07}.
Note that $\SttracesM(\specTA)$ only contains the quiescence label $\delta$ at \emph{exactly} $M$ time units.
\begin{definition}[$\tiocoM$]
\label{def:tioco}
Let $\implTA$ be an IOTA, $\specTA$ be a TA and $M \in \RRplus$, then we define $\implTA\ \tiocoM\ \specTA$ iff
$$
\forall\ \ttrace \in \SttracesM (\specTA) : \outM(\implTA \afterM \ttrace) \subseteq \outM(\specTA \afterM \ttrace).
$$
\end{definition}
\section{Transformations}
\label{sec:transformation}
We show that \ioco\ is preserved when going from  LTS/IOTS to TA/IOTA.
Below we define the transformation from an LTS to a TA.
\autoref{def:transformation} is central to our contribution: it details the conversion of an LTS into a TA in which quiescence is represented by a dedicated transition labelled~$\delta$ that becomes enabled after a time-out~$M$, whereas other output transitions are restricted to the $t < M$.

Introducing the notion of time naturally implies the addition of clocks.
Since LTSs do not inherently have a notion of time, the only clock $\clock$ that is necessary to add is the one to measure quiescence.
The parameter $M\in\RRplus$ is the explicit time when we declare a location as quiescent.
Naturally, each output transition gets a guard $\clock<M$ to ensure it is taken strictly before $M$ time units have passed.
The only enabled output at exactly $M$ time units is the quiescent transition $\delta$.
To ensure that either an output or $\delta$ is observed, either before $M$ time units have passed (when an output is observed), or when exactly $M$ time units have been observed (when quiescence is observed), we add the invariant $c \le M$ to each location.
Inputs also need to be provided strictly before $M$, because it is not needed to wait for quiescence before providing the input.
Inputs, outputs and quiescence represent the \emph{visible} behaviour of the system, and quiescence measures the time passed since the last visible progress.
Thus, like a stopwatch, the clock $\clock$ is reset on every visible action.
This closely reflects how testing of timed systems is done in practice.
Our work provides the mathematical underpinning of its correctness.

\begin{definition}[TA-ification]
\label{def:transformation}
    The TA-ification of an LTS,\\ $\system=\ltstuple$, for parameter $M\in\RRplus$ is a function $\taification$:
    $\mathit{LTS}\rightarrow\mathit{TA}$, with $\systemTA = \langle \locations, \Actq, \PhiLoc, \{\clock\}, \rightarrow_\sysTA, \state_0 \rangle$
    such that: 
    \begin{enumerate}
        \item States $S$, including initial state $s_0$, identify locations of $\systemTA$
        \item $\Actq$ are the actions labels of $\systemTA$ (i.e. all actions of $\system$ and additionally $\delta$) 
        \item $\clock$ is the unique clock of $\systemTA$ used to track quiescence      
        \item $\PhiLoc:\locations\rightarrow\Phi(\clocks)$ is a function assigning clock constraints to $\systemTA$'s locations as follows:       
        $\PhiLoc(\location) = \clock \leq M$
        
        \item $\vaa_{\sysTA}$ defines $\sysTA's$ transition relation as an extension of $\vaa$ with clock constraints and resets, as follows:
        \begin{align*}
        \vaa_{\sysTA}\ = 
            \{(\location,\action,\{\clock < M\},\{\clock\},\location') &\mid (\location,\action,\location') \in\  \vaa \cap\ (\locations \times \Act \times \locations) \}\ \cup\\
            \{(\location, \delta,\{\clock = M\},\{\clock\},\location) &\mid \location \in \locations \text{ is quiescent}\}
        \end{align*}
    \end{enumerate}
    For an LTS $\system$ we call its resulting TA, i.e. $\taification(\system) = \sysTA$, the \emph{canonic} TA of $\system$.
\end{definition}

We depict the TA-ification in \autoref{tab:taification}. It shows the TA-ification of a single transition, disregarding the other transitions of its source location.
\begin{table}[ht!]
    \renewcommand{\arraystretch}{2} 
    \setlength{\tabcolsep}{4pt}
    \centering
    \begin{tabular}{| >{\centering\arraybackslash}m{0.17\textwidth}|
                    >{\centering\arraybackslash}m{0.24\textwidth}|
                    >{\centering\arraybackslash}m{0.24\textwidth}|
                    >{\centering\arraybackslash}m{0.24\textwidth}|}
        \hline
        \textbf{LTS Transition} 
        & \tikztrans{$\iaction?$}{\phantom{guard}}{\phantom{reset}} 
        & \tikztrans{$\oaction!$}{\phantom{guard}}{\phantom{reset}} 
        & \tikztrans{$\delta$}{\phantom{guard}}{\phantom{reset}} \\
         \hline
        \textbf{TA transition after $\taification$} 
        & \tikztransinv{$\iaction?$}{$\clock < M$}{$\{\clock\}$}{$\clock \le M$}
        & \tikztransinv{$\oaction!$}{$\clock < M$}{$\{\clock\}$}{$\clock \le M$}
        & \tikztransinv{$\delta$}{$\clock = M$}{$\{\clock\}$}{$\clock \le M$} \\
        \hline
    \end{tabular}
\bigskip
    
    \caption{
    Visual representation of $\taification$ (cf. Definition~\ref{def:transformation}). 
    All locations enabling output (including quiescence $\delta$) get an invariant and their transitions get an appropriate guard to enforce that $\delta$ can \emph{only} be observed after $M$ time units.}
    \label{tab:taification}
\end{table}

\begin{example}
    Consider the LTS $\system$ and its corresponding TA $\taification(\system)$ presented in Figure~\ref{fig:lts-ta-transformation}.
    After the transformation from states to locations, each 
    state with an output-outgoing transitions was decorated with a $\clock \leq M$ invariant and all output-outgoing transitions are annotated with a $\clock<M$ guard. 
    In this manner we ensure that any visible output strictly occurs before $M$ time units. 
    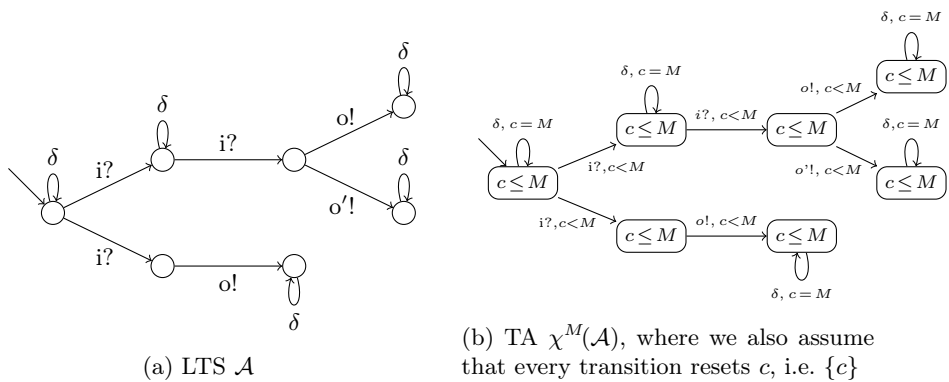
\begin{figure}[ht!]
    \hspace{-1.5cm}
    \centering
    \begin{subfigure}[b]{0.45\textwidth}
    \begin{tikzpicture}
    \tikzset{state/.style= {draw, circle}} 
    \tikzset{timedstate/.style= {draw, rounded corners}}
    \node[state] (s0) {};
    \node[above left of=s0] (init) {};
    \node[state, below right of=s0, xshift=0.75cm] (s1) {};
        \node[state, right of=s1, xshift=0.75cm] (s3) {};
    \node[state, above right of=s0, xshift=0.75cm] (s2) {};
        \node[state, right of=s2, xshift=0.75cm] (s4) {};
            \node[state, below right of=s4, xshift=0.75cm] (s5) {};
            \node[state, above right of=s4, xshift=0.75cm] (s6) {};

    \draw[->] (init) to (s0);
    \draw[->] (s0) to node[below left, xshift=0.2cm] {\iaction?} (s1);
    \draw[->] (s0) to node[above left, xshift=0.2cm] {\iaction?} (s2);
    \draw[->] (s1) to node[below] {\oaction!} (s3);
    \draw[->] (s2) to node[above] {\iaction?} (s4);
    \draw[->] (s4) to node[below left, xshift=0.2cm] {\ooaction!}(s5);
    \draw[->] (s4) to node[above left, xshift=0.2cm] {\oaction!} (s6);
    
    \draw[->] (s0) to [loop above] node {$\delta$} (s0);
    \draw[->] (s2) to [loop above] node {$\delta$} (s2);
    \draw[->] (s3) to [loop below] node {$\delta$} (s3);
    \draw[->] (s5) to [loop above] node {$\delta$} (s5);
    \draw[->] (s6) to [loop above] node {$\delta$} (s6);
    \end{tikzpicture}

    \subcaption{LTS $\system$}
    \end{subfigure}
    \quad\quad
    \begin{subfigure}[b]{0.45\textwidth}
    \begin{tikzpicture}
    \tikzset{state/.style= {draw, circle}} 
    \tikzset{timedstate/.style= {draw, rounded corners}}
    \node[timedstate, below right of=s0, xshift=0.5cm] (s0) {\scriptsize{$\clock\!\leq\!M$}};
    \node[above left of=s0] (init) {};
    
    \node[timedstate, below right of=s0, xshift=1cm] (s1) {\scriptsize{$\clock\!\leq\!M$}};
    \node[timedstate, right of=s1, xshift=1cm] (s3) {\scriptsize{$\clock\!\leq\!M$}};
    \node[timedstate, above right of=s0, xshift=1cm] (s2) {\scriptsize{$\clock\!\leq\!M$}};
    \node[timedstate, right of=s2, xshift=1cm] (s4) {\scriptsize{$\clock\!\leq\!M$}};
    \node[timedstate, below right of=s4, xshift=0.75cm] (s5) {\scriptsize{$\clock\!\leq\!M$}};
    \node[timedstate, above right of=s4, xshift=0.75cm] (s6) {\scriptsize{$\clock\!\leq\!M$}};

    \draw[->] (init) to (s0);
    \draw[->] (s0) to node[below left, xshift=0.25cm] {\tiny{\iaction?,$\clock\!\!<\!\!M$}} (s1);
    \draw[->] (s0) to node[above left, xshift=0.9cm, yshift=-0.35cm] {\tiny{\iaction?,$\clock\!\!<\!\!M$}} (s2);
    \draw[->] (s1) to node[above] {\tiny{$\oaction!,\clock\!\!<\!\!M$}} (s3);
    \draw[->] (s2) to node[above] {\tiny{$\iaction?,\clock\!\!<\!\!M$}} (s4);
    \draw[->] (s4) to node[below left, xshift=0.2cm] {\tiny{$\ooaction!,\clock\!\!<\!\!M$}}(s5);
    \draw[->] (s4) to node[above left, xshift=0.2cm] {\tiny{$\oaction!,\clock\!\!<\!\!M$}} (s6);
    
    \draw[->] (s0) to [loop above] node {\tiny{$\delta,\clock\!=\!M$}} (s0);
    \draw[->] (s2) to [loop above] node {\tiny{$\delta,\clock\!=\!M$}} (s2);
    \draw[->] (s3) to [loop below] node {\tiny{$\delta,\clock\!=\!M$}} (s3);
    \draw[->] (s5) to [loop above, yshift=-0.2cm] node {\tiny{$\delta,\!\clock\!=\!M$}} (s5);
    \draw[->] (s6) to [loop above] node {\tiny{$\delta,\clock\!=\!M$}} (s6);
    \end{tikzpicture}
    \subcaption{TA $\taification(\system)$, where we also assume that every transition resets $\clock$, i.e. $\{\clock\}$}
    \end{subfigure}
    \caption{Transformation of an LTS $\system$ into a TA with a TA-ification $\taification(\system)$
 }
    \label{fig:lts-ta-transformation}
\end{figure}

\end{example}

For proving conformance preservation of our $\taification$ operator, we translate back from $\taification(\system)$ to $\system$. 
When we formally want to disregard the timing aspects (i.e. guards, invariants, clocks) it is useful to formally define a TA's \emph{projection} onto its untimed system.
\begin{definition}[Projection]
\label{def:projection}
      Let $\specTA=\tatuple$ be a TA. Then its \emph{projected LTS} is: $\spec = \langle \locations,  \Act, \rightarrow', \locations_0 \rangle$, where: 
      
      $$\rightarrow' = \{\langle \ell,a,\ell'\rangle\in S\times\Act\times S \mid \langle \ell, a, \phi, \lambda, \ell'\rangle \in \rightarrow\}$$

      Let $\ttrace=(d_1,\action_1)\ldots(d_n,\action_n) \in \Sttraces(\specTA)$ be a suspended timed trace.
      We define the \emph{projection} of $\sigma$, i.e. its regular suspended trace without time, as $\project{\ttrace}=\action_1\ldots \action_n$.
    
\end{definition}

The following lemma describes the relation between traces of an LTS and timed traces in its canonic TA.
In essence, we may disregard the time-label associated to each action since the transformation $\taification$ neither adds nor removes behaviour, e.g. we associate $(d_1,\action_1)\ldots(d_n,\action_n)$ with $\action_1\ldots\action_n$ for some $d_i\in\RRplus$.
\begin{lemma}[Canonic Traces]
\label{lem:trace-trafo}
    Let $\system=\ltstuple$ be an LTS and $M\in\RRplus$, then:
    \begin{enumerate}
        \item \label{lem-bullet:LTStoTA} If $\trace\in\Straces(\sys)$, then there is $\ttrace\in\SttracesM(\taification(\sys))$ such that $\project{\ttrace}=\trace$
        %\MS{You mean $\ttrace\in\traces(\qsys)$, right?}
        \item \label{lem-bullet:TAtoLTS} If $\ttrace\in\SttracesM(\taification(\sys))$, then there is $\trace\in\Straces(\sys)$ such that $\project{\ttrace}=\trace$.
    \end{enumerate}
\end{lemma}
With Lemma~\ref{lem:trace-trafo} every original transition of an LTS is preserved under the transformation to a TA. Hence, it easily follows that the transformation also preserves input--enabledness.
\begin{corollary}
\label{cor:IOTA}
Let $\sys$ be an IOTS and $M\in\RRplus$, then $\taification(\sys)$ is an IOTA.    
\end{corollary}
We are now able to state one of our paper’s main contributions: transformation of an LTS to a TA preserves conformance from \ioco to $\tiocoM$.
\begin{theorem}
\label{thm:iocopreservation}
Let $\impl$ be an IOTS and $\spec$ be an LTS. Then:
$$
\impl\ \iocorel\ \spec \Longleftrightarrow  \taification(\impl)\ \tiocorel_M\ \taification(\spec)
$$
for all $M\in\RRplus$.
\end{theorem}
Theorem~\ref{thm:iocopreservation} guarantees preservation of conformance, allowing practitioners to model only the system’s input–output behaviour.
Quiescence and time-outs are added explicitly by the transformation described in \autoref{def:transformation}.

Technically, Theorem~\ref{thm:iocopreservation} works for any non-zero time-out $M$. In reality, practitioners should use the time-out time that works best in their domain.
\section{Testing with TAs}
\label{sec:testing}
In this section we investigate the practical half of our testing theory. First we define test cases for LTSs and TAs, respectively.
Our test cases are inspired from the literature in~\cite{T2008,PvdB18} for the LTS case, and~\cite{BBB04} for the TA case.

We present our core result that mirrors the practical side of \ioco being preserved under the LTS-to-TA-transformation: 
Concretely, if an implementation passes every untimed test of its LTS specification, it will also pass every timed test of the lifted specification and conversely, any untimed test that shows non-conformance has a timed counterpart that reveals the same defect. 
Additionally, we show that $\taification$ commutes with test generation: applying $\taification$ before or after the standard ioco
test-generation algorithm yields the same set of tests.

\subsection{Test Cases for LTS}

We model test cases for an LTS $\system$ as tree-shaped LTSs with the same inputs (representing what a tester may send) and same outputs (what the SUT may emit), and explicit $\delta$ actions in its transition relation.
The intuition is that every state presents a choice for the tester: Either (1) stop -- decide the verdict and end the test, (2) observe -- wait to see whether the SUT produces an output or remains quiescent, or (3) stimulate -- send an input to the SUT.

The last option comes with one caveat: while we are about to apply an input, the SUT is still allowed to ``interrupt'' us with an output. To handle races like this cleanly, the test tree always enables the entire output set alongside the chosen input. Notably, quiescence is \emph{not} enabled in such a state, since we want to apply the input before the (as of yet time-agnostic) quiescence-timeout.

\begin{definition}[Test case for LTS]
\label{def:ltstest}
A test case for an LTS $\spec$ is an LTS $\test=\langle \states^\test, \qAct, \trans^\test, \startingstate^\test \rangle$ s.t.:
 \begin{itemize}
    \item $\test$ uses the same action labels as $\spec$ plus $\delta$
        \item $\test$ has only finite traces, is deterministic and has no cycles 
    \item There are two special states \pass, \fail $\in \states^\test$
    \item States {\pass} and {\fail} have no outgoing transitions:  \\  
        $\forall\ \action \in \Actq: \pass\ntransition{\action} \y\ \fail\ntransition{\action}$        
    \item Every other state enables all outputs $\Acto$, and either one input  or $\delta$, i.e.\\
        $\forall\ \state \in \states^\test\setminus\{\pass,\fail\}:$\\
        $ (|\inp(\state)|\!=\!0 \wedge \out(\state)\!=\!\Actoq) \vee (\out(\state)\!=\!\Acto \wedge |\inp(\state)|\!=\!1)$
    \item Input-specifiedness: All traces of $\test$ that end with an input are suspended traces of $\spec$, i.e. 
        $\forall\ \trace \in (\qAct)^*: \forall\ \iaction? \in \Acti: \trace\!\cdot\!\iaction? \in \traces(\test) \Trans \trace\!\cdot\!\iaction? \in \Straces(\spec)$
    \item Soundness: All traces of $\test$ leading to {\pass}, are suspended traces of $\spec$\\
        $\forall\ \trace \in \traces(\test): \test \transition{\trace} \pass \Trans \trace \in \Straces(\spec)$ 
    \item Correctness: All traces of $\test$ that end with an output and lead to {\fail}, are not suspended traces of $\spec$:\\ 
        $\forall\ \trace \in (\qAct)^*,\forall\ \oaction! \in \Acto:$\\
        $\trace\!\cdot\!\oaction! \in \traces(\test) \y \test\transition{\trace\!\cdot\!\oaction!}\fail \Trans \trace\!\cdot\!\oaction! \not\in \Straces(\spec).$
 \end{itemize}
\end{definition}

Running a test case $\test$ against an implementation IOTS $\impl$ may yield different traces, due to the presence of nondeterministic choices in $\impl$.
Then $\test$ fails on $\impl$ if at least one of the traces of $\impl$ leads to a fail-verdict in $\test$. 

\begin{definition}[Test verdict]
\label{def:ltstestverdict}
Let $\test$ be a  test case for an LTS $\spec$, and let $\impl$ be an implementation IOTS, 
\begin{align*}
\text{$\impl$ fails $\test \iff                                     
        \exists\ \trace\in\Straces(\impl) \cap \traces(\test): \test\va{\sigma}\fail$} 
\end{align*}
Likewise, $\impl$ passes $\test$ iff $\impl$ \cancel{fails} $\test$.
\end{definition}

We note that we do not consider asynchronous inconsistencies in communication between the tester and the SUT, where input sending from the test case and output observation from the SUT may conflict with each other \cite{PvdB18}. In this paper, we abstract from this and just deal with the traces that have been executed \cite{T2008}.%, both for the untimed test cases of this subsection and timed test cases of next subsection. 
We could extend this in future work along the lines of \cite{PvdB18}; our paper is then still valid, since the conflict resolution yields the trace that has actually been executed.

\subsection{Test Cases for TA}
We model \emph{timed} test cases for a TA as tree-shaped TAs that reuse the same inputs, outputs and explicit action $\delta$. All intuitions from the LTS test case model carry over, but two timed-specific tweaks matter:
(1) Time can elapse while we \emph{observe}. In every node of the TA time may now elapse as long as the state invariant holds. Quiescence is detected when \emph{no} output occurs during such a delay. The concrete quiescence-timeout bound is encoded as a guard on the transition. 
(2) Races now involve outputs \emph{and} time. When we decide to \emph{stimulate} the system, we may do so with some delay smaller than the quiescence delay. This makes the race between the SUT emitting output and us providing input explicit.
We note that we restrict our timed test cases to \emph{canonic} TAs of LTSs, because this excludes all timed automata with multiple clocks, or non-trivial invariants and guards.

\begin{definition}[Test case for TA]
\label{def:tatest}
A test case $\testTA$ for $\specTA$ is a $\testTA=\langle \locations^\test, \qAct, \Phi^\test_{\locations^\test}, \calC^\test, \rightarrow^\test, \location^\test_0\rangle$ such that:
 \begin{itemize}
    \item $\testTA$ is a canonic TA for a given $\spec$ and $M\in \RRplus$
    \item $\testTA$ uses the same action labels as $\specTA$ plus $\delta$
    \item There are two special locations \pass, \fail $\in \locations^\test$
    \item Locations {\pass} and {\fail} have no outgoing transitions: \\ 
        $\forall\ \action \in \Actq, \forall\ \tim \in \RRo: \pass\nva{(\tim,\action)} \y\ \fail \nva{(\tim,\action)} $ 
    \item $\testTA$ has only finite traces and has no cycles 
    \item $\testTA$ is deterministic, i.e. \\ 
    $\forall\ \action \in \Actq, \forall\ \tim \in \RRo, \forall\ \location \in \locations^\test: |\location \afterM (d,\action)| \leq 1$%(\exists!\ \location' \in \locations^\test: \location \va{(\tim,\action)} \location') \vee \location \not\va{(\tim,\action)}$  
    \item Every location enables except {\pass} and {\fail} all outputs $\Acto$, and either one input  or $\delta$, i.e.
    
        $\forall\ \location \in \locations^\test\setminus\{\pass,\fail\}: \\
        (|\inp(\location)|\!=\!0 \wedge \out(\ell)=\!\Actoq) \vee (\out(\location)\!=\!\Acto \wedge |\inp(\location)|\!=\!1)$
    \item All locations except {\pass} and {\fail} have the invariant $\clock \le M$, i.e. \\ 
        $\forall\ \location \in \locations^\test\setminus\{\pass,\fail\}:  \PhiLoc^\test(\location) = (\clock \le M)$
    \item All non-$\delta$ transitions have clock guard $\clock < M$, i.e.\\ 
        $\forall\  \langle \location, \action, \phi, \lambda, \location'\rangle \in\ \rightarrow^\test: \action \neq \delta \implies \phi = (\clock < M)$
    \item All $\delta$ transitions have clock guard $\clock = M$, i.e.\\ 
        $\forall\ \langle \location, \action, \phi, \lambda, \location'\rangle \in\ \rightarrow^\test: \action = \delta \implies \phi = (\clock = M)$

    \item Input-specifiedness: all traces of $\testTA$ that end with an input are suspended traces of $\specTA$, i.e.\\ 
        $\forall\ \ttrace \in \ttraces(\testTA), \forall\ \iaction? \in \Acti, \forall\ d \in \RRnonzero:\\ d \leq M \wedge \ttrace\,\cdot\,(d,\iaction?) \in \ttraces(\testTA) \Trans \ttrace\,\cdot\,(d,\iaction?) \in \SttracesM(\specTA)$
    \item Soundness: All timed traces of $\testTA$ leading to {\pass}, are suspended timed traces of $\specTA$\\
        $\forall\ \ttrace \in \ttraces(\testTA): \testTA \transition{\ttrace} \pass \Trans \ttrace \in \SttracesM(\specTA)$ 
    \item Correctness: All timed traces of $\testTA$ that end with an output and lead to {\fail}, are no suspended timed traces of $\specTA$:\\ 
        $\forall\ \ttrace\!\cdot\!(\tim,\oaction!) \in \ttraces(\testTA), \forall\ \oaction! \in \Acto, \forall\ \tim \in \RRnonzero: d \leq M \wedge \testTA\transition{\ttrace\!\cdot\!(\tim,\oaction!)}\fail\ \Trans \ttrace\!\cdot\!(\tim,\oaction!) \not\in \SttracesM(\specTA).$

  \end{itemize}
\end{definition}

Similarly, running a test case $\test$ against an implementation IOTS $\impl$ may yield different traces, due to the presence of nondeterministic choices in $\impl$.
Then $\test$ fails on $\impl$ if at least one of the traces of $\impl$ leads to a fail-verdict in $\test$.

\begin{definition}[Timed test verdict]
\label{def:tatestverdict}
Let $\testTA$ be a  test case for a TA $\specTA$, and let $\implTA$ be an implementation IOTA. We say:
\begin{align*}
\text{$\implTA$ fails $\testTA \iff \exists\ \ttrace \in                      \SttracesM(\implTA)\cap\ttraces(\testTA):\testTA \transition{\ttrace} \fail$.} 
\end{align*}
Likewise, $\implTA$ passes $\testTA$ iff $\implTA$ \cancel{fails} $\testTA$.
\end{definition}

\begin{example}
For the LTS $\system$ in \autoref{fig:lts-ta-transformation}(a), a corresponding test case is shown in \autoref{fig:test-lts-ta}(a).  
Similarly, \autoref{fig:test-lts-ta}(b) gives a test for the TA $\systemTA$ of \autoref{fig:lts-ta-transformation}(b).
\begin{figure}
    \hspace{-1cm}
    \centering
    \begin{subfigure}[b]{0.45\textwidth}
    \begin{tikzpicture}
    \tikzset{state/.style= {draw, circle}} 
    \tikzset{timedstate/.style= {draw, rounded corners}}
    \node[state] (s0) {};
    \node[above left of=s0] (init) {};

    \node[state, above right of=s0, xshift=0.75cm] (s1) {};
    \node[timedstate, right of=s0, xshift=0.75cm] (s2) {\scriptsize{\fail}};
    \node[timedstate, below right of=s0, xshift=0.75cm] (s3) {\scriptsize{\fail}};

    \node[timedstate, above right of=s1, xshift=0.75cm] (s4) {\scriptsize{\pass}};
    \node[timedstate, right of=s1, xshift=0.75cm] (s5) {\scriptsize{\pass}};
    \node[timedstate, below right of=s1, xshift=0.75cm] (s6) {\scriptsize{\fail}};

    \draw[->] (init) to (s0);
    \draw[->] (s0) to node[above] {\iaction?} (s1);
    \draw[->] (s0) to node[above] {\oaction!} (s2);
    \draw[->] (s0) to node[above, xshift=0.4cm, yshift=-0.2cm] {\ooaction!} (s3);
    
    \draw[->] (s1) to node[above] {$\delta$} (s4);
    \draw[->] (s1) to node[above] {\oaction!}(s5);
    \draw[->] (s1) to node[above, xshift=0.4cm, yshift=-0.2cm] {\ooaction!} (s6);
    
    \end{tikzpicture}
    \subcaption{A test case for $\system$}
    \end{subfigure}
    \quad\quad
    \begin{subfigure}[b]{0.45\textwidth}
    \begin{tikzpicture}
    \tikzset{state/.style= {draw, circle}} 
    \tikzset{timedstate/.style= {draw, rounded corners}}
    \node[timedstate, below right of=s0, xshift=0.5cm] (s0) {\scriptsize{$\clock\!\leq\!M$}};
    \node[above left of=s0] (init) {};
    
    \node[timedstate, above right of=s0, xshift=1cm] (s1) {\scriptsize{$\clock\!\leq\!M$}};
    \node[timedstate, right of=s0, xshift=1cm] (s2) {\scriptsize{\fail}};
    \node[timedstate, below right of=s0, xshift=1cm] (s3) {\scriptsize{\fail}};
    
    \node[timedstate, above right of=s1, xshift=1cm] (s4) {\scriptsize{\pass}};
    \node[timedstate, right of=s1, xshift=1.4cm] (s5) {\scriptsize{\fail}};
    \node[timedstate, below right of=s1, xshift=1cm] (s6) {\scriptsize{\fail}};

    \draw[->] (init) to (s0);
    \draw[->] (s0) to node[above, xshift=-0.25cm] {\tiny{$\iaction?,\!\clock\!\!<\!\!M,\!\clock$}} (s1);
    \draw[->] (s0) to node[above, yshift=-0.07cm] {\tiny{$\oaction!,\!\clock\!\!<\!\!M,\!\clock$}} (s2);
    \draw[->] (s0) to node[below, xshift=-0.35cm, yshift=0.05cm] {\tiny{$\ooaction!,\!\clock\!\!<\!\!M,\!\clock$}} (s3);
    
    \draw[->] (s1) to node[above, xshift=-0.25cm, yshift=-0.05cm] {\tiny{$\delta,\!\clock\!\!=\!\!M,\!\clock$}} (s4);
    \draw[->] (s1) to node[above, yshift=-0.07cm] {\tiny{$\oaction!,\!\clock\!\!<\!\!M,\!\clock$}} (s5);
    \draw[->] (s1) to node[above, xshift=0.45cm, yshift=-0.09cm] {\tiny{$\ooaction!,\!\clock\!\!<\!\!M,\!\clock$}}(s6);

    \end{tikzpicture}
    \subcaption{A test case for $\systemTA = \taification(\system)$}
    \end{subfigure}
    \caption{Corresponding test cases before and after $\taification$}
    \label{fig:test-lts-ta}
\end{figure}
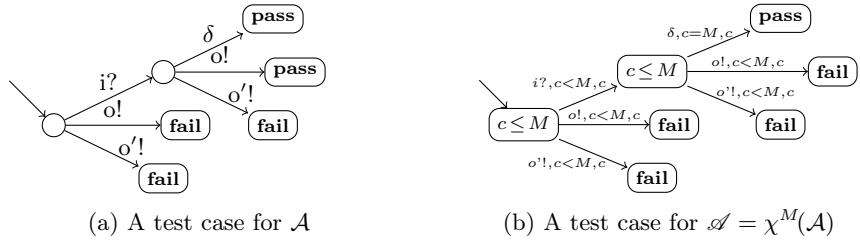

\end{example}

Every trace with at least one output that is not included in the specification is labelled fail.
This aligns with the concept of underspecifications in \ioco.
In practice, we generate test cases on-the-fly according to the specification model.

\subsection{Testing}

The transformation in~\autoref{def:transformation} allows the specification activity to remain entirely in the untimed setting. A modeller provides a plain LTS that captures the functional behaviour of the intended system. From here \emph{all} auxiliary machinery-- adding quiescence and having a uniform time-out bound is introduced automatically by $\taification$. Hence, the effort of dealing with time and quiescence is shifted from the modeller to the transformation.

Given~\autoref{def:ltstest} and~\autoref{def:tatest}, the next theorem establishes a one-to-one correspondence between the test suites obtained in the untimed and in the timed paradigm.

\begin{theorem}[Test Correspondence]
\label{thm:testcorrespondence}
    Let $\spec$ be an LTS and let $\atestLTS(\spec)$ be the set of all its tests according to~\autoref{def:ltstest}. Similarly, let $\atestTA(\taification(\spec))$ be the set of all tests for $\taification(\spec)$, which is a TA according to~\autoref{def:tatest}. Then, for all $M\in\RRplus:$
    $$\taification(\atestLTS(\spec)) = \atestTA(\taification(\spec))$$
\end{theorem}
\smallskip

\autoref{thm:testcorrespondence} shows that test derivation and the transformation $\taification$ commute: transforming a test after it is generated yields exactly the same result as generating the test after the specification has been transformed. This commutation property and~\autoref{thm:testsoundness} are the technical core of our presented work, because both guarantee that every verdict reached in the untimed paradigm is mirrored in the timed one. We therefore lift correctness properties from $\iocorel$ to $\tiocoM$.
\begin{theorem}
\label{thm:testsoundness}
Let $M\in\RRplus$, $\impl$ be an IOTS and $\spec$ be an LTS. 
Let $\atestLTS(\spec)$ and $\atestTA(\taification(\spec))$ be the set of all annotated tests for $\spec$ and $\taification(\spec)$, respectively. Then:
\begin{enumerate}\itemsep0mm
    \item \label{thm:1-sound}
    If $\impl\ \text{passes } \atestLTS(\spec)$, then $\taification(\impl)\ \text{passes } \atestTA(\taification(\impl))$

    \item\label{thm:2-comp}
    If $ \impl\ \text{fails } \atestLTS(\spec)$, then $\taification(\impl)\ \text{fails } \atestTA(\taification(\impl)).$ 

\end{enumerate}
\end{theorem}

\section{Related Work}
\label{sec:related-work}

\paragraph{Timed conformance.} 
The testing of real-time systems has long been a topic of research interest. 
Early work extended the theory of testing deterministic Mealy machines to timed input/output automata~\cite{SVD01}. 
Later, \cite{NS03} expand these results by considering a determinizable class of non-deterministic timed systems.
From there, much of the theory evolved into LTS‑based ioco frameworks~\cite{T1996}.
This is because LTSs capture non‑deterministic branching while maintaining an operational view of inputs and outputs.
The best‑known timed variations are tioco~\cite{LMN2004}, rtioco~\cite{KT2004}, and the variant adopted here, $\tiocoM$~\cite{BBB04}. 
Also of note are dtioco for distributed systems and multi-traces~\cite{GHG13} and live timed ioco (ltioco), which further distinguishes two flavours of quiescence and works directly on TA zone graphs~\cite{LGL19}.
Other work combines timed properties with probabilities~\cite{GHS19,N11}.

Timed testing remains a topic of interest, even beyond ioco and LTSs, e.g.~\cite{AKY24} use timed FSMs, and~\cite{CSP23} augment CSP with discrete time. The work of~\cite{K21} brings forth decidability results for TAs with one clock--which relates to our systems after the transformation $\taification$.

\paragraph{Quiescence and ioco.}
Tretmans~\cite{T1996} introduced quiescence in his seminal work on ioco theory via suspension traces, which include $\delta$ to denote the absence of observable output. The concept was later refined, with Stokkink et al.\ treating $\delta$ as a first-class citizen and discussing well-formedness rules~\cite{T2008} and divergence~\cite{STS2012}. 

Numerous ioco variants have been developed to suit different modelling contexts and application domains, including compositional ioco, i.e. uioco \cite{JT2019,TJ2022}, probabilistic ioco~\cite{GS18}, symbolic ioco~\cite{FrantzenTW06,BosT19}, and modal ioco~\cite{LochauPKS14}.

\paragraph{Model transformations and tool support.}
To operationalise conformance, design artefacts must first be transformed into testable models, then automated tools can execute the resulting test suites and provide verdicts.
Two such related model-to-test transformations are provided by~\cite{PF99}, who take the opposite route to ours and translate a TA into an untimed one to leverage the arsenal of available untimed testing techniques, and~\cite{ITV20} who derive timed test cases directly from UML activity diagrams.
Noteworthy tools of industrial maturity realise these theories, e.g. UPPAAL Tron~\cite{LMN2004} executes tioco test suites and RT-Tester~\cite{PH16} supports safety-critical certifications in the automotive and aviation industry.

\section{Conclusion}
\label{sec:conclusion}
We provided a lightweight route for lifting untimed LTS models and test suites into the timed domain.
Central to our approach is the canonic TA-ification operator $\taification$ (cf. Definition~\ref{def:transformation}), which augments any LTS/IOTS with a single clock that models quiescence explicitly as a time-out at a user-chosen bound~$M$.
We provide proofs that $\taification$ preserves conformance from the classical \ioco\ relation to $\tiocoM$ (cf. Theorem~\ref{thm:iocopreservation}), and, via a tight construction of test cases (cf.~\autoref{def:ltstest} and~\autoref{def:tatest})~\autoref{thm:testsoundness} guarantees that every fail detectable in the untimed setting remains detectable once timing constraints are introduced via the transformation. 
Additionally, we showed that $\taification$ commutes with test generation: applying $\taification$ before or after the standard ioco
test-generation algorithm yields the same set of tests.

Overall, our work lowers the entry barrier for timed conformance testing: existing LTS specifications, implementation models and off-the-shelf ioco tooling can be used.   
This is how most ioco-based testing with quiescence was done in practice regardless--we provided the formal underpinning enabling this practice.
An intriguing next step in this line of work is integrating support for internal actions ($\tau$-actions). The authors of \cite{STS2013} show how divergence can be explicitly modelled as quiescence.
A complementary research direction from the perspective of a practitioner is the automatic inference for optimal time-out bounds $M$, for example from observed traces or domain knowledge, which would reduce manual tuning even further.

\subsubsection{Acknowledgments:} %This work was supported by the European Union’s Horizon 2020 research and innovation programme under Marie Skłodowska-Curie grant agreement 101008233 (\textsc{MISSION}).
This project has received funding from the European Union’s Horizon 2020 research and innovation programme under the Marie Skłodowska-Curie grant agreement No 101008233 (\textsc{MISSION}).

%- - - - - Bibliography - - - - - 
\bibliographystyle{abbrv}
\bibliography{thebib.bib}
%- - - - - - - - - - - - - - - - -

%- - - - - Appendices - - - - - - -
\newpage
\appendix
\section{Omitted Proofs}
Below we provide the proofs for the main results in our paper.
\setcounter{lemma}{\getrefnumber{lem:trace-trafo}-1}
\begin{lemma}[Canonic Traces]
    Let $\system=\ltstuple$ be an LTS and $M\in\RRplus$, then:
    \begin{enumerate}
        \item If $\trace\in\Straces(\sys)$, then there is $\ttrace\in\SttracesM(\taification(\sys))$ such that $\project{\ttrace}=\trace$.
        %\MS{You mean $\ttrace\in\traces(\qsys)$, right?}
        \item If $\ttrace\in\SttracesM(\taification(\sys))$, then there is $\trace\in\Straces(\sys)$ such that $\project{\ttrace}=\trace$.
    \end{enumerate}
\end{lemma}
\begin{proof}
%The proof is done in two steps:

\boxed{Case\ \ref{lem-bullet:LTStoTA}} If $\trace\in\Straces(\sys)$ we need to show that for given $M\in\RRplus$ there exists $\ttrace\in\SttracesM(\taification(\sys))$ such that $\project{\ttrace}=\trace$ (cf.~\autoref{def:projection}). The proof is by induction on the trace length $\lvert\trace\rvert$ of $\trace$. 

\emph{Base case.} Consider the empty trace $\trace=\varepsilon$ with trace length zero, i.e. $\lvert\trace\rvert=0$. In an LTS, the empty trace means that the system stays in the initial state $\startingstate$. With~\autoref{def:transformation} $\taification(\system)$ is a timed automaton with initial location $\location_0$. With no transition taken, the projection of any empty (suspended) timed trace $\ttrace\in\SttracesM(\taification(\system))$ is also empty, i.e. $\project{\ttrace}=\varepsilon=\trace$. Thus, the claim holds for $\lvert\trace\rvert$.

\emph{Induction hypothesis.} Assume that the claim holds for traces of length $n$ for some $n\in\mathbb{N}$. We now show that the claim also holds for traces of length $n+1$.

\emph{Induction Step.}
Assume $\trace'\in\Straces(\sys)$ with $\lvert\trace'\rvert=n+1$.
Thus, there are $\trace\in\Straces(\sys)$ with $\lvert\trace\rvert=n$ and $\action\in\qAct$ such that we may also write $\trace'=\trace\cdot \action$.
By definition of traces (cf.~\autoref{def:lts})---and by trivial extension to suspended traces---there are states $\state_i\in \states$ such that:
$$\startingstate\transition{\action_1}\state_1\transition{\action_2}\ldots\transition{\action_n}\state_n\transition{\action}\state_{n+1}, \text{or equivalently}$$
$$\startingstate\transition{\trace}\state_n\transition{\action}\state_{n+1},  \text{ with } \trace = \action_1 \cdots \action_n.$$

% Apply Hypothesis
According to the induction hypothesis, there is a (suspended) timed trace $\ttrace\in\SttracesM(\taification(\system))$ with $\project{\ttrace}=\trace$. Moreover we know $\location_n \afterM \ttrace$ as the equivalent location in $\taification(\system)$ to the state $s_n$ (since each $s_i$ identifies $\ell_i$ directly).
From here it suffices to show that for the last transition of the trace in the LTS, i.e. $(\state_n,\action,\state_{n+1})\in\ \trans_{\sys}$ (or $(\state_n,\delta,\state_n)$ resp.), there is a transition in the timed automaton, i.e. $(\location_n,\action,g,\{\clock\},\location_{n+1})\in\ \trans_{\taification(\sys)}$ for some guard $g\in\Phi$ and clock reset $\clock \subseteq\clocks$.
We distinguish the three cases: $\action=\iaction?\in\Acti$, $\action=\oaction!\in\Acto$ and $\delta$:
%\MG{Is it clear that $c_n\approx \location_n$? | Moreover, there is $\location_n$ for each $s_n$ 'aka add that states and labels are interchangeable'}

\begin{description}
    \item[$\boxed{a=\iaction?}$] According to the transformation (cf.~\autoref{def:transformation}) this means that there is a TA transition $(\location_n,\iaction?,c<M,\{\clock\},\location_{n+1})$. Also, in $s_n$ there is either
    \begin{enumerate} 
        \item at least one outgoing output $\oaction!$, or there is
        \item no outgoing  output, in which case the state is quiescent and a $\delta$ is added in the suspension traces. 
        
    \end{enumerate}
    In either case, due to the transformation rules (cf.~\autoref{def:transformation}) there is a clock invariant $\clock\leq M$ in $\location_n$.
    This means that, for all $0\leq d<M$ there is a transition: 
    $$\location_n\transition{(d,\iaction?)}\location_{n+1}$$
    
    % \MG{This is interesting. It's ok here, because we're effecttively looking for any d, not all d.}
    %
    \item[$\boxed{a=\oaction!}$] With~\autoref{def:transformation}  there is a TA transition $(\location_n,\oaction!,\clock<M,\{\clock\},\location_{n+1})$. In particular, there is a clock invariant $\clock\leq M$ in $\location_n$. This means that, for all $0\leq d<M$ there is a transition $$\location_n\transition{(d,\oaction!)}\location_{n+1}$$
    \item[$\boxed{\delta}$] With~\autoref{def:transformation} there is a TA transition $(\location_n,\delta,\clock=M,\{\clock\},\location_{n+1})$. In particular, there is a clock invariant $\clock\leq M$ in $\location_n$. This means there is a transition $$\location_n\transition{(M,\delta)}\location_{n+1}$$
\end{description}

In each case (depending on input, output or $\delta)$ we can choose  $0\leq d\leq M$ such that $\ttrace'=\ttrace\cdot(\tim,\action)\in\SttracesM(\taification(\sys))$. Moreover $\project{\ttrace'}=\trace'$ by construction, which concludes the induction.

\boxed{Case\ \ref{lem-bullet:TAtoLTS}} 
The proof is by construction. Let $\ttrace\in\SttracesM(\taification(\sys))$. Then we need to find $\trace\in\Straces(\sys)$ such that $\project{\ttrace}=\trace$.
With~\autoref{def:tanotation} any suspension timed traces can be written as:
$$\location_1 \transition{(d_1,\action_1)}\location_2 \transition{(d_2,\action_2)}\ldots\transition{(d_{n-1},\action_{n-1})}\ell_n$$
In particular, with the definition of the transformation (cf.~\autoref{def:transformation}) each transition in $\trans_{\taification(\sys)}$ originates from an LTS transition, or is an explicitly added $\delta$ self-loop in quiescent states. The TA-ification only extends the original discrete transitions with additional guards and resets.
More precisely, by~\autoref{def:transformation} every transition in $\trans_{\taification(\sys)}$ is one of the three forms for $\action= \iaction?$, $\action= \oaction!$ and $\action= \delta$:
\begin{itemize}
    \item $(\location, \iaction?, c<M, \{\clock\},\location')$ for inputs  $(\state, \iaction?, \state')\in\ \trans_{\sys}$ and $\iaction?\in\Acti$,
    \item $(\location, \oaction!, \clock< M, \{\clock\},\location')$ for outputs  $(\state, \oaction!, \state')\in\ \trans_{\sys}$ and $\oaction!\in\Acto$,
    \item $(\location, \delta, \clock=M, \{\clock\},\location')$ for quiescent states, i.e. $(\state,\delta,\state)$ in suspended traces.
\end{itemize}
In each case, according to~\autoref{def:projection}, the projection $\project{\ttrace}$ removes the information of the transition that is only relevant for TAs and not in LTSs (i.e. delays $d\in\RRnonzero$), and we are left with the suspended trace
$$\trace= \action_1 \action_2 \ldots\action_n\in\Straces(\sys)$$
Consequently, $\project{\ttrace}=\trace$, which concludes the proof.

\qed
\end{proof}
\setcounter{theorem}{\getrefnumber{thm:iocopreservation}-1}
\begin{theorem}
Let $\impl$ be an IOTS and $\spec$ be an LTS. Then:
$$
\impl\ \iocorel\ \spec \Longleftrightarrow  \taification(\impl)\ \tiocorel_M\ \taification(\spec)
$$
for all $M\in\RRplus$.
\end{theorem}
% \MS{I think the converse also holds, so the thm can be written with $\iff$}
% \MG{This proof uses suspension timed traces, otherwise fine.}
% TODO: The other direction should be straightforward
%
\begin{proof}
    %
    % Base assumption
%    \begin{itemize}
    \boxed{\implies}
    %\item[$\boxed{\oaction!\in\Acto}$] 
    Let $\impl$ be an IOTS, $\spec$ be an LTS and assume that $\impl\ \iocorel\ \spec$. 
    %
    % Needs to be proven
    Let $M\in\RRplus$, then we need to show that $\taification(\impl)\ \tiocorel\ \taification(\spec)$.
    %
    % Applying Definition of tioco
    According to~\autoref{def:tioco}, we need to show that for all suspension timed traces of the specification $\ttrace \in \SttracesM(\taification(\spec))$ it holds that
    $$
     \outM(\taification(\impl)\afterM \ttrace) \subseteq \outM(\taification(\spec)\afterM \ttrace)
    $$
    %
    % Applying premise for proof 
    Thus, let $\ttrace \in \SttracesM(\taification(\spec))$.
    %
    % Case distinction: Does output exist or not?
    If $\ttrace\notin\SttracesM(\taification(\impl))$ the inclusion is trivial because then $\outM(\taification(\impl)\afterM \ttrace)=\varnothing$. Otherwise, pick $(\tim,\oaction)\in \outM(\taification(\impl)\afterM \ttrace)$ with $d\in\RRplus$ and $\oaction\in\olabsd$. 
    % Assuming $\oaction\in\olabs$ again leads to $\outM(\taification(\impl)\afterM \ttrace)=\varnothing$ by definition of $\mathbf{out}_M$ given in Remark~\ref{rem:tioconotation}. 
    Then $\ttrace\!\cdot\!(d,\oaction)=\ttrace'$ and consequently $\ttrace'\in\Sttraces(\taification(\impl))$.
    %
    % Specyfing what needs to be proven
    It remains to be shown that $\ttrace'\in\SttracesM(\taification(\spec))$ as this guarantees that $(d,\oaction)\in\outM(\taification(\spec)\afterM \ttrace)$.
    
    %
    % Applying transformation definition
    With~\autoref{lem:trace-trafo} and the definition of the transformation (cf.~\autoref{def:transformation}), there is an untimed trace $\trace\in\Straces(\impl)$ such that $\project{\ttrace}=\trace$.
    By the same argument we can derive $\trace'\in\Straces(\impl)$ from $\ttrace'$, i.e.
    %
    % Using construction from above
    with $\oaction\in\Actoq$ and $\trace'=\trace\!\cdot\!\oaction$ we again use the transformation (cf.~\autoref{def:transformation}) and~\autoref{lem:trace-trafo} to conclude that $\trace'=\trace\cdot \oaction\in\Straces(\impl)$.    
    %
    % Applying the ioco premise for LTS
    Based on the premise $\impl\ \iocorel\ \spec$ we know that
    $$
    \forall\ \xi \in \Straces(\spec):
    \out(\impl\after \xi) \subseteq \out(\spec\after \xi)
    $$
    and thus with $\oaction\in\out(\impl\ \after \trace)\subseteq\out(\spec\ \after \trace)$ also $\trace'\in\Straces(\spec)$.
    
    %
    % Using the transformation definition again, this time LTS -> TA
    %Using~\autoref{lem:trace-trafo} and the transformation in~\autoref{def:transformation} again we can conclude that $\oaction$ is indeed an output after $\trace$ and 
    We are left to determine $d\in\RRnonzero$.
    %
    % Case distinction in b\in\Acto versus \oaction=\delta
    For that we distinguish two cases: $\oaction=\delta$ and $\oaction\in\Acto$. 
    In the last transition of $\ttrace'$ there are $\location,\location'\in\locations$ such that $(\location,\oaction,\phi,\calC,\location')$ for some guard $\phi$ and clock resets $\calC$.
    % Without loss of generality we may neglect any preceding or proceeding $\tau$ transitions. 
    % Qua Definition~\ref{def:transformation} their guard set requires the unique clock to have valuation $0$ and the clock reset set is empty. 
    % Any $\tau$ progress is assumed to happen instantaneously, thus prohibiting any clock progression.
    % Moreover, only $\tau$ transitions that are present in the LTS before the transformation are preserved, thus no new behaviour is added.
    
    %
    % Handling first case
    \begin{itemize}
    \item Assume $\oaction=\delta$.  
    By~\autoref{def:transformation} this implies both $I(\location)=\{\clock\leq M\}$ and $\phi=(\clock=M)$.
    Thus, the only available trace in $\SttracesM(\taification(\impl))$ requires $d=M$.
    With~\autoref{rem:tioconotation} this implies $(M,\delta)\in\outM(\taification(\impl) \afterM \ttrace)$.
    %With~\autoref{def:transformation} and the premise $\qimpl\ \ioco\ \qspec$ 
    We conclude $(M,\delta)\in\outM(\taification(\spec) \afterM \ttrace)$.
    % 
    % Handling second case
    \item The case of $\oaction!\in\Acto$ proceeds analogously, albeit with $\phi=(\clock<M)$ and some $\tim<M$, again via the transformation rules of~\autoref{def:transformation}.
    \end{itemize}
    
    %
    % Wrapping everything up
    Summarizing both cases we conclude that there is a duration $d\leq M$ such that $(d,\oaction)\in\outM(\taification(\spec) \afterM \ttrace)$.
    Since $M\in\RRplus$ was arbitrary but fixed, this shows that for all $\ttrace\in\SttracesM(\taification(\spec))$:
    $$
    \outM(\taification(\impl) \afterM \ttrace)\subseteq \outM(\taification(\spec)\ \afterM \ttrace).
    $$
    In addition to this, we know via~\autoref{cor:IOTA}  that $\taification(\impl)$ is an IOTA, which finally lets us conclude $\taification(\impl)\ \tiocorel_M\ \taification(\spec)$.

    \boxed{\Longleftarrow} Assume that $\taification(\impl)\ \tiocorel_M\ \taification(\spec)$ with $M\in\RRplus$. 
    This means that for all $\ttrace \in \SttracesM(\taification(\spec))$ and for all $(d, \oaction) \in \outM(\taification(\impl) \afterM \ttrace)$ we know that $(d, \oaction) \in \outM(\taification(\spec) \afterM \ttrace)$.
    
    To show that $\impl\ \iocorel\ \spec$ according to~\autoref{def:ioco}, we need to show that for all suspension traces of the specification $\trace \in \Straces(\spec)$ it holds that
    $$\out(\impl \after \trace) \subseteq \out(\spec \after \trace).$$
    For that we will use induction.
    
    \emph{Base case.} Let $\trace \in \Straces(\spec)$. If $\trace = \varepsilon$ with need to prove that for all $\oaction \in \out(\impl \after \varepsilon)$ it holds that $\oaction \in \out(\spec \after \varepsilon)$. With~\autoref{lem:trace-trafo} and $\taification(\impl)\ \tiocorel_M\ \taification(\spec)$ this follows directly since $\project{\varepsilon}=\varepsilon$. 

    \emph{Induction hypothesis.} Assume that for all timed traces $\ttrace$ with length $n$ it holds that if $\oaction! \in \out(\impl \after \trace)$ then $\oaction! \in \out(\spec \after \trace)$. We continue to prove that for all timed traces $\trace' = \action \cdot \trace$ with length $n+1$ it holds that if $\oaction! \in \out(\impl \after \trace')$ then $\oaction! \in \out(\spec \after \trace')$.

    \emph{Induction step.} Suppose there is an output such that $\oaction! \in \out(\impl \after \trace')$ and $\oaction! \not\in \out(\spec \after \trace')$. With~\autoref{lem:trace-trafo} this means that exists a $d\in\RRnonzero$ such that $(d, \oaction) \in \outM(\taification(\impl) \afterM \ttrace')$ and $(d, \oaction) \not\in \outM(\taification(\spec) \afterM \ttrace')$ where $\project{\ttrace'}=\trace'$. However, this cannot happen because  $\taification(\impl)\ \tiocorel_M\ \taification(\spec)$. This implies for all output such that $\oaction! \in \out(\impl \after \trace')$ we know $\oaction! \in \out(\spec \after \trace')$, which concludes the induction and the proof. 
    \qed
   % \end{itemize}
\end{proof}
\setcounter{theorem}{\getrefnumber{thm:testcorrespondence}-1}
\begin{theorem}[Test Correspondence]
    Let $\spec$ be an LTS and let $\atestLTS(\spec)$ be the set of all its tests according to Definition~\ref{def:ltstest}. Similarly, let $\atestTA(\taification(\spec))$ be the set of all tests for $\taification(\spec)$ which is a TA according to Definition~\ref{def:tatest}. Then, for all $M\in\RRplus$:
    $$\taification(\atestLTS(\spec)) = \atestTA(\taification(\spec))$$
\end{theorem}
\begin{proof}
Let $M\in\RRplus$. The proof is done in two steps:
\item[$\boxed{\taification(\atestLTS(\spec)) \subseteq \atestTA(\taification(\spec))}$]
% "Take an LTS test, transform it and show that the properties hold"
Let $\test\in\atestLTS(\spec)$ be an LTS test case for $\spec$ in correspondence with~\autoref{def:ltstest}. We must show that $\exists\ \testTA \in \atestTA(\taification(\spec)): \taification(\test)=\testTA$ satisfies every clause of the TA-test definition (\autoref{def:tatest}).
We go step-by-step below:
\begin{description}

    % Projection onto LTS and inherited structural properties
    % \item[Structure.] We must show that the projected test case (cf.~\autoref{def:projection} of $\testTA$ is a test case for the LTS $\spec$, i.e. $\project{\testTA}\in\atestLTS$. This is clear by construction of $\taification$ (cf. Definition~\ref{def:transformation}), Lemma~\ref{lem:trace-trafo} and since $\test$ is a test according to Definition~\ref{def:ltstest}. Removing all clocks, guards, resets and time-delays in $\testTA$ we recover the original LTS $\test$. 
    % Moreover, the transformed test $\testTA$ is still deterministic, acyclic (apart from loops in $\pass$ and $\fail$ states), input-specified, sound and correct, because no new behaviour was added by merit of Lemma~\ref{lem:trace-trafo}.
    % \MG{Arguably this is a bit handwavy, but I think this suffices here.}
    % \MG{Not sure if we defined 'projected test case' yet.}
    \item[Structure.] The transformation $\taification$ performs a structure-preserving relabelling (cf.~\autoref{def:transformation},~\autoref{lem:trace-trafo}). It keeps the graph structure, i.e. nodes, quiescent states and transitions, adds one clock and decorates transitions and locations. Therefore, $\pass/\fail$ nodes, the tree shape and determinism are preserved. 
    
    % % Invariants in locations
    \item[Invariants.] We show that for all locations $\ell\in L$ except $\pass$ and $\fail$
    it holds that 
    $\Phi_L(\ell)=\{\clock\leq M\}$. By~\autoref{def:ltstest} $\test$ either has exactly one outgoing input transition or exactly no  outgoing input transition, i.e.
    $$\forall\ \ell \in \locations: (|\inp(\ell)|\!=\!0 \wedge \out(\ell)\!=\!\Acto \cup \{\delta\}) \vee (\out(\ell)\!=\!\Acto \wedge |\inp(\ell)|\!=\!1)$$
    %Assume it has no input transition leaving location $\ell$, then $\out(\ell)=\Acto\cup\{\delta\}$. According to the transformation (Definition~\ref{def:transformation}) then $\ell$ has an invariant $\clock\leq M$.
    %Conversely, assuming there is exactly one input transition leaving $\ell$, then $\out(\ell)=\Acto$. 
    Irrespective of $\out(\ell)\!=\!\Acto\cup\{\delta\}$ or $\out(\ell)\!=\!\Acto$, according to the transformation (\autoref{def:transformation}) $\ell$ has the invariant $\clock\leq M$.
    % Guards on transitions
    \item[Guards.] We show that the guard sets on transitions of $\testTA$ conform to the ones required of~\autoref{def:tatest} by making a distinction between inputs, non-$\delta$ outputs and $\delta$ outputs. According to ~\autoref{def:transformation} :
        \begin{itemize}  
        \item ... every input transition in $\testTA$ is of the form $\langle\ell,\iaction?,c<M,\{c\},\ell'\rangle$,
        %\MG{There is a discrepancy we need to fix. We can't define the transformation such that input guards are always $\true$ but in tests they can only go up to $k\leq M$.}
    
        \item ... every non-$\delta$ output transition in $\testTA$ is of the form $\langle\ell,\oaction!,c < M,\{c\},\ell'\rangle$,
    
        \item ... every $\delta$ output transition in $\testTA$ is of the form $\langle\ell,\delta,c=M,\{c\},\ell'\rangle$.
        \end{itemize}
        In all three cases this is exactly what is required in~\autoref{def:tatest}. This means that all guard sets on transitions of $\testTA$ conform to~\autoref{def:tatest}.
    \item[Input-specifiedness, soundness, correctness.] Similar to the preservation in structure, $\taification$ neither adds nor deletes traces (cf.~\autoref{lem:trace-trafo}), so the three properties can be lifted unchanged.
\end{description}
All properties combined yield $\taification(\test)=\testTA$. 
    
\item[$\boxed{\taification(\atestLTS(\spec)) \supseteq \atestTA(\taification(\spec))}$]  
Let $\testTA\in\atestTA(\taification(\spec))$. We need to show that there is a test $\test\in\atestLTS(\spec)$, such that $\testTA=\taification(\test)$.
For that, we consider the projection of the test case $\testTA$ (cf.~\autoref{def:transformation}) as a candidate, i.e. we consider $\project{\testTA}$ and show that it is the test case we are looking for.
Clearly $\project{\test}=\langle \locations,\Actq,\rightarrow',\ell_0\rangle$ where 
$$\rightarrow' = \{(\ell, \action, \ell')\in L\times\Actq\times L\mid (\ell,\action,\phi,\{\clock\},\ell'\in\rightarrow_{\testTA})\}.$$
% ! !
% The way we define projections, this actually "carries over" the delta-label to the projected LTS meaning that in this specific instance, delta will actually be a real label in the LTS. 
%
The projection is trace preserving and does not add new behaviour (cf.~\autoref{lem:trace-trafo}). Specifically, $\delta$ labels are only explicitly added in quiescent states during the transformation, meaning that their LTS counterpart has a corresponding $\delta$ in the suspension traces. The same holds when we apply the transformation of~\autoref{def:transformation}.
%\MG{TODO: This should be obvious; Revisit and work out more if time allows.}
Structurally, timed tests are conservative extensions of regular LTS test cases (cf.~\autoref{def:ltstest} and~\autoref{def:tatest}), hence we conclude $\test=\project{\testTA}\in\testLTS$.
What remains to be shown is $\testTA=\taification(\test)$. For that we show that the location invariants and the transition guards are the same.
\begin{description}
    \item[Invariants.] According to~\autoref{def:transformation} all locations have the clock invariant $\Phi=\{\clock\leq M\}$. The same is true for timed test cases (cf.~\autoref{def:tatest}).
    \item[Guards.] Like before, the case distinction is between inputs, non-$\delta$ outputs and $\delta$ outputs. We observe that the transition guards in both~\autoref{def:transformation} and~\autoref{def:tatest} are
    \begin{itemize}
        \item $\phi=(\clock < M)$ for inputs
        \item $\phi=(\clock < M)$ for non-$\delta$ outputs and
        \item $\phi=(\clock = M)$ for $\delta$ outputs
    \end{itemize}
    This implies that the transition guards after transformation $\taification$ and those of timed test cases in $\atestTA$ are the same.
\end{description}
Ultimately this yields $\testTA=\taification(\project{\testTA})$, and with it $\testTA\subseteq\taification(\testLTS(\spec))$. 

\qed
\end{proof}
\setcounter{theorem}{\getrefnumber{thm:testsoundness}-1}
\begin{theorem}
Let $M\in\RRplus$, $\impl$ be an IOTS and $\spec$ be an LTS. 
Let $\atestLTS(\spec)$ and $\atestTA(\taification(\spec))$ be the set of all annotated tests for $\spec$ and $\taification(\spec)$, respectively. Then:
\begin{enumerate}\itemsep0mm
    \item
    If $\impl\ \text{passes } \atestLTS(\spec)$, then $\taification(\impl)\ \text{passes } \atestTA(\taification(\spec))$

    \item
    If $ \impl\ \text{fails } \atestLTS(\spec)$, then $\taification(\impl)\ \text{fails } \atestTA(\taification(\spec)).$ 

\end{enumerate}
\end{theorem}
    \begin{proof}
    Let $\impl,\spec,M,\atestLTS(\spec)$ and $\atestTA(\taification(\spec))$ be as specified.
    
    \begin{itemize}
        \item[\boxed{\ref{thm:1-sound}}] 
        % How are we proving? Contraposition
        The proof is via contraposition, i.e. we prove if $\taification(\impl)\ \text{fails } \atestTA(\taification(\spec))$, then it follows that $\impl\ \text{fails } \atestLTS(\spec)$. 
        
        % Assumption
        Thus, assume $\taification(\impl)\ \text{fails } \atestTA(\taification(\spec))$. 
        % Unrolling definition
        By~\autoref{def:tatestverdict} this implies that there is a timed test case $\testTA\in\atestTA(\taification(\spec))$ such that there is a (suspended) timed trace that leads to a fail state, i.e.
        $$ \exists\ \ttrace\in\SttracesM(\taification(\impl))\cap\ttracesM(\testTA):\testTA\vva{\ttrace} \fail$$
        For this timed trace $\ttrace$ let $\trace=\project{\ttrace}$ be its projected trace (cf.~\autoref{def:projection}).
        From here we deduce two properties: 
        \begin{enumerate}
            \item Since $\ttrace\in\SttracesM(\taification(\impl))$ we have $\trace\in\Straces(\impl)$ (cf.~\autoref{lem:trace-trafo}), and 
            \item By the definitions of test cases (cf.~\autoref{def:ltstest}) and timed test cases (cf.~\autoref{def:tatest}) we know that this trace leads to a fail state in the untimed test $\test=\project{\testTA}$, i.e. $\sigma\in\traces(\test) \text{ and } \test\vva{\sigma}\fail. $
        \end{enumerate}
        % Concluding
        Moreover, with~\autoref{thm:testcorrespondence}, $\project{\testTA}\in\atestLTS$. This means we found an untimed test case $\test\in\atestLTS(\spec)$ that contains a (suspended) trace $\sigma$ which leads $\impl$ to a fail state. By~\autoref{def:ltstest} this means that $\impl\ \text{fails } \atestLTS$.%, as was to be proven.
        \item[\boxed{\ref{thm:2-comp}}] 
        % Assumption and verdict definition
        Assume $\impl\ \text{fails } \atestLTS$, then according to~\autoref{def:ltstestverdict} there is $\atest\in\atestLTS(\spec)$ for which there is $\trace\in\Straces(\impl)\cap\traces(\test)$ with $\test\vva{\trace}\fail$.
        
        With~\autoref{thm:testcorrespondence} we know that $\taification(\test)\in\atestTA(\taification(\spec))$. Likewise, with~\autoref{lem:trace-trafo} we know there is $\ttrace\in\SttracesM(\taification(\impl))\cap\ttraces(\taification(\test))$ such that $\project{\ttrace}=\trace$.

        Notably, with the definition of the transformation (cf.~\autoref{def:transformation}) this means we found $$\ttrace\in\SttracesM(\taification(\impl))\cap\ttraces(\testTA):\testTA\vva{\ttrace}\fail.$$
        
        % Stating the goal
        With~\autoref{def:tatestverdict} this means $\taification(\impl)\ \text{fails } \testTA$ and with it $\taification(\impl)$ fails $\atestTA(\taification(\spec))$.

    \end{itemize}
    \qed
\end{proof}
% \begin{center}
%     $\forall\ \impl \in LTS : \impl\ \pass\ \testLTS$ then $\taification(\impl)\ \pass\ \testTA$ \\

%     If $\exists\ \impl \in LTS : \impl\ \npass\ \testLTS$ then $\taification(\impl)\ \pass\ \testTA$ \\

%     $\taification(\testLTS) \neq \testTA$ moreover $\taification(\testLTS)$ include $\testTA$ \\

%- - - - - - - - - - - - - - - - - 
\end{document}